\documentclass[12pt]{amsproc}
\relax
\usepackage[mathscr]{eucal}
\usepackage{amsmath,amsfonts,amssymb,amsthm}
\usepackage{times}
\usepackage{hyperref}

 \newtheorem{prop}{Proposition}[section]
 
 \newtheorem{definition}[prop]{Definition}

\renewcommand{\hat}{\widehat}

\newcommand{\bref}[1]{\textbf{\ref{#1}}}

\newcommand{\p}[1]{|#1|}
\newcommand{\gh}[1]{\mathrm{gh}(#1)}

\newcommand{\dv}{\mathrm{d_v}}
\newcommand{\dx}{\mathrm{d_X}}

\renewcommand{\d}{\partial}
\renewcommand{\dh}{\mathrm{d_h}}

\newcommand{\binner}[2]{%
  {\langle}\kern-4.15pt{\langle}#1{,}\,#2{\rangle}\kern-4.15pt{\rangle}}
\newcommand{\commut}[2]{[#1{,}\,#2]}

\newcommand{\pb}[2]{\left\{{}#1{},{}#2{}\right\}}
\newcommand{\ab}[2]{\big(#1,#2\big)}

\newcommand{\half}{\mathchoice{%
    \ffrac{1}{2}}{\frac{1}{2}}{\frac{1}{2}}{\frac{1}{2}}}

\newcommand{\ffrac}[2]{\raisebox{.5pt}%
  {\footnotesize$\displaystyle\frac{#1}{#2}$}\kern1pt}

\newcommand{\tot}{\mathrm{tot}}

\newcommand{\dl}[1]{\mathchoice{\ffrac{\d}{\d #1}}{\frac{\d}{\d #1}}{\ffrac{\d}{\d #1}}{\ffrac{\d}{\d #1}}}

\newcommand{\st}[2]{{\overset{#1}{#2}}}

\newcommand{\Liealg}{\mathfrak} 
\newcommand{\algg}{\Liealg{g}}

\newcommand{\CC}{\mathcal{C}}
\newcommand{\cC}{\mathcal{C}}

\newcommand{\fR}{\mathbb{R}}

\newcommand{\fZ}{\mathbb{Z}}

 \def\cE{\mathcal{E}}
 \def\cG{\mathcal{G}}
 
 \def\cI{\mathcal{I}}
 
\def\cK{\mathcal{K}}
 \def\cL{\mathcal{L}}
 \def\cP{\mathcal{P}}

\def\tr{{\rm Tr}}

\usepackage[left=2.5cm,right=2.5cm,top=3cm,bottom=3cm,bindingoffset=0cm]{geometry}

\newcommand{\hL}{\cL}
\newcommand{\Ruth}[1]{#1}
\newcommand{\pr}{\mathrm{pr}}

\begin{document}
\title{Presymplectic gauge PDEs and Lagrangian BV formalism beyond jet-bundles \vspace{0.5cm}}

\author{Maxim Grigoriev\vspace{0.2cm}}
\address{Lebedev Physical Institute,
  Leninsky Ave. 53, 119991 Moscow, Russia \vspace{1em}
  \\
  \newline
  \phantom{aa~}
  Institute for Theoretical and Mathematical Physics, \newline
  \phantom{aa~} Lomonosov Moscow State University, 119991 Moscow, Russia}  

  \subjclass[2020]{
  Primary: 35,70,49; Secondary: 53, 83}

  \email{grig@lpi.ru}

\date{~}

\numberwithin{equation}{section}

\begin{abstract}
A gauge PDE is a geometrical object underlying what physicists call a local gauge field theory
defined at the level of equations of motion (i.e. without specifying Lagrangian) in terms of Batalin-Vilkovisky (BV) formalism. This notion extends the BV formulation in terms of jet-bundles on the one hand and the geometrical approach to PDEs on the other hand. In this work we concentrate on gauge PDEs equipped with a compatible presymplectic structure and show that under some regularity conditions this data defines a jet-bundle BV formulation. More precisely, the BV jet-bundle arises as the symplectic quotient of the super jet-bundle of the initial gauge PDE. In this sense, presymplectic gauge PDEs give an invariant geometrical approach to Lagrangian gauge systems, which is not limited to jet-bundles.  Furthermore, the presymplectic gauge PDE structure naturally descends to space-time submanifolds (in particular, boundaries, if any) and, in this respect, is quite similar to AKSZ sigma models which are long known to have this feature. We also introduce a notion of a weak presymplectic gauge PDE, where the nilpotency of the differential is replaced by a presymplectic analog of the BV master equation, and show that it still defines a local BV system. This allows one to encode BV systems in terms of finite-dimensional graded geometry, much like the AKSZ construction does in the case of topological models.
\vspace{0.2cm}
\end{abstract}

\maketitle
\tableofcontents
\section{Introduction}

Modern theories of fundamental interactions such as  Yang-Mills theory, Einstein gravity, and their modifications are local gauge field theories or, at least, can be described by such theories in certain regimes. 

At the classical level such systems are partial differential equations (PDE), which in addition are required to be variational (Lagrangian) and to possess a family of variational symmetries locally parametrized by unconstrained space-time functions (gauge parameters). A powerful method of handling such systems is the Batalin-Vilkovisky (BV) formalism~\cite{Batalin:1981jr,Batalin:1983wj,Batalin:1984ss} (see also \cite{Henneaux:1992ig,Gomis:1995he} for a review) that extends usual Lagrangian formalism by incorporating the structure of the gauge symmetries into the extended configuration space which, in turn, naturally happens to be a differential graded symplectic manifold. 

The version of the BV approach, which explicitly takes care of space-time locality, is also well known by now, see e.g.~\cite{Barnich:1995db,Barnich:1995ap} and \cite{Barnich:2000zw} for a review and \cite{Barnich:2004cr,Barnich:2010sw,Kaparulin:2011xy,Sharapov:2016sgx} for more recent developments. In this case, the BV configuration space is a jet-bundle, and in addition to the BV extension of the Lagrangian form and the BV symplectic structure, it comes equipped with the canonical Cartan distribution. This data also gives rise to a natural cohomology theory, the so-called local BRST cohomology, which happens to encode the invariant properties of the system and has direct physical meaning (global symmetries, consistent interactions, anomalies, counterterms etc.).

Although, in the geometrical theory of PDEs, the jet-bundle approach appears to be quite useful, it is not sufficiently flexible and is not invariant in the sense that one and the same PDE can be realised in terms of different jet-bundles. An invariant approach was developed by Vinogradov and his school and it deals with the infinitely-prolonged equations seen
(in the simplest version) as bundles endowed with a Cartan distribution~\cite{Vinogradov:1977,Vinogradov:1978,Vinogradov:1984}, see also~\cite{Krasil'shchik:2010ij,Krasil?shchik-Lychagin-Vinogradov} for a review. 

Similar ideas can be applied to
local gauge theories described in terms of BV formalism. First of all, forgetting about symplectic structure gives a version of the jet-bundle BV approach applicable to not necessarily variational PDEs with a gauge structure~\cite{Barnich:2004cr,Lyakhovich:2004xd}. Furthermore, such a BV formulation can be extended beyond jet-bundles, giving a more flexible and invariant approach~\cite{Barnich:2010sw}. This approach was made more general and geometrical in~\cite{Grigoriev:2019ojp}, where the respective geometrical object was called a gauge PDE and its relation to the Vinogradov's notion of diffiety was made explicit. In particular, if a usual PDE is considered as a gauge PDE then the respective diffiety is obtained as its minimal model, i.e. by maximal equivalent reduction of the initial gauge PDE.

The notion of gauge PDE is not a straightforward merging of the notion of PDE and BV system  (that would be just a PDE in the differential-graded category; this is known as a standard gauge PDE~\cite{Grigoriev:2019ojp}). Instead, the Cartan distribution is somehow united with the gauge distribution encoded in the analog of the BV-BRST differential so that there is only one $\fZ$-degree and one differential $Q$. This generalization turns out to be crucial in describing the frame-like formulation of gauge theories or, equivalently, gauge theories based on Cartan geometry. In particular, this approach includes, as a special case, so-called unfolded formalism~\cite{Vasiliev:1988xc,Vasiliev:2005zu} which is based on representing equations of motion as a free differential algebra~\cite{Sullivan:1977fk}. In fact, the main ingredients of the gauge PDE construction were initially proposed in~\cite{Barnich:2004cr,Barnich:2010sw} as a tool to generate unfolded formulations and to relate them to the BV-BRST approach. In this context it is also worth mentioning the group-geometric approach to supergravity~\cite{DAuria:1982nx}, which also employs free differential algebras.

In field theory applications we are mostly interested in Lagrangian systems. However, by its very definition Lagrangian is defined on a jet-bundle and hence cannot be directly encoded in the gauge PDE geometry. A natural question is to extend the gauge PDE approach to Lagrangian systems. It turns out that Lagrangian can often be (but not always is) encoded in the compatible graded presymplectic structure defined on the gauge PDE. In the case of a usual variational PDE, this presymplectic structure is precisely the presymplectic structure defined on the equation manifold~\cite{Kijowski:1979dj,Crnkovic:1986ex,Zuckerman:1989cx,Andersonbook} (see also \cite{Khavkine2012,Sharapov:2016qne,Druzhkov:2021} for more recent developments relevant in the present context). More precisely, given an equation manifold equipped with a compatible presymplectic structure, it defines a natural action functional~\cite{Grigoriev:2016wmk,Grigoriev:2021wgw}, the so-called intrinsic action, defined on sections of the equation manifold (seen as a bundle over the space-time manifold). The structure of the intrinsic action originates from that of  presymplectic AKSZ sigma-models~\cite{Alkalaev:2013hta} which, in turn, generalize 
 AKSZ sigma models~\cite{Alexandrov:1995kv} (see also~\cite{Cattaneo:1999fm,Grigoriev:1999qz,Batalin:2001fc,Cattaneo:2001ys,Roytenberg:2002nu,
  Bonechi:2009kx,Barnich:2009jy,Bonavolonta:2013mza,Ikeda:2012pv} for further developments), initially proposed as an elegant supergeometrical construction of the BV formulation of topological models. 
  
The relation to AKSZ approach suggests that the intrinsic action construction can be extended to gauge PDEs. This was shown to be indeed the case in~\cite{Grigoriev:2020xec} in the case of Einstein gravity but the construction extends to generic diffeomorphism-invariant gauge theories. Extension to general gauge PDEs was recently proposed in~\cite{Dneprov:2022jyn}. It is also worth mentioning recent application of presymplectic AKSZ sigma models to constructing Lagrangians for higher spin gravity~\cite{Sharapov:2021drr}.

In this work we study presymplectic gauge PDEs, i.e. gauge PDEs equipped with compatible presymplectic structure. We apply the efficient technique of super jet-bundles to prove that a presymplectic gauge PDE indeed gives rise to a local BV system. Moreover, it defines a tower of symplectic structures of descending degree and a tower of compatible charges of which the top one is a Lagrangian itself. This is analogous to the case of AKSZ sigma models. In particular these symplectic structures and charges can be integrated over the submanifolds of the respective dimension. This shows that just like AKSZ models presymplectic gauge PDEs naturally induce the analogous structure on space-time submanifolds and/or boundaries.

We also introduce a concept of weak presymplectic gauge PDE, which is obtained by replacing the nilpoltency condition $Q^2=0$ with $\omega(Q,Q)=0$ condition. Remarkably, this does not affect the entire construction and, again, the local BV system naturally emerges on a suitable symplectic quotient. This generalization appears useful in applications  because it often allows one to replace genuine gauge PDE with a finite-dimensional weak analog which still encodes all the information on the system. The example of Yang-Mills theory in this formulation is explicitly constructed.

The paper is organized as follows: in Section~\bref{sec:prelim} we recall how gauge systems are described in terms of graded geometry. This discussion is extended to local gauge theories in Section~\bref{sec:gPDE}. This is done in terms of so-called gauge PDEs, which are also briefly recalled. The structure of presymplectic gauge PDEs is studied in Section~\bref{sec:presymp-gauge-PDE}. In Section~\bref{sec:weak} we introduce a weak analog of presymplectic gauge PDE and show that it still encodes a local BV system. The construction is illustrated with the example of Yang-Mills theory. Finally, in Section~\bref{sect:boundary} we demostrate that presymplectic gauge PDE structure naturally descents to space-time submanifolds and boundaries and again consider Yang-Mills theory as an example.

\section{Preliminaries}
\label{sec:prelim}

\subsection{$Q$-manifolds}
What we are after is a geometric setup underlying classical gauge systems. If we temporarily disregard additional structures such as space-time locality and Lagrangian, the data of a gauge system consists of a manifold of field configurations together with a submanifold singled out by the equations of motion, a gauge distribution that is involutive on the submanifold, and the data defining gauge for gauge symmetries and their higher analogs.  All these can be conveniently packed into the structure of $Q$-manifold:
\begin{definition}
$Q$-manifold $(M,Q)$ is a $\fZ$-graded supermanifold equipped with a homological vector field $Q$, i.e. a vector field of degree $1$ satisfying $Q^2=0$, $\gh{Q}=1$, $\p{Q}=1$, where $\gh{\cdot}$ denotes $\fZ$-degree (often called ghost degree) and $\p{\cdot}$ Grassmann parity. 
\end{definition}
It is probably more appropriate to call it dg manifold  but we follow the tradition initiated in~\cite{Schwarz:1992gs,Schwarz:1992nx}.  For simplicity,
we restrict ourselves to the case where fermionic fields and/or fermionic gauge parameters are not present so that $\fZ_2$-degree  (Grassmann parity) is determined by $\fZ$-degree, i.e. given by $\fZ$-degree mod 2.

\subsection{$Q$-manifolds as gauge systems}

It is instructive to see how the information about the underlying gauge system can be extracted from the data of a $Q$-manifold. To be more precise, by a gauge system we temporarily mean a system whose space-time manifold is 0-dimensional. The considerations still apply to generic systems but then one should disregard all the space-time related issues such as locality, space-time geometry, etc. and is forced to work with, in general,  infinite-dimensional manifolds of field configurations.

Given a $Q$-manifold $(M,Q)$, its body (which can be seen as a submanifold obtained by setting to zero all the coordinates of nonvanishing  degree) is a space of "field configurations". The submanifold singled out by the equations of motion (its points are "solutions" of the underlying system) is the intersection of the zero-locus of $Q$ with the body of $M$. It is useful to give a more invariant characterization. Namely, the body of $M$ can be seen as the space of maps from a point (considered as a $Q$-manifold with the trivial $\fZ$-grading and the trivial $Q$-structure) to $M$. Then solutions (points where $Q$-vanishes) can be identified as $Q$-maps, i.e. maps satisfying $\sigma^* \circ Q=0$. In the same language the gauge transformations (the  gauge distribution) can be defined as follows:
$\delta_\xi \sigma^*=\xi_\sigma^* Q$, where $\xi_\sigma^*$ is a "gauge parameter map" $\cC^\infty (M) \to \fR$ satisfying $\gh{\xi_\sigma^*}=-1$
and $\xi_\sigma^*(fg)=\xi_\sigma^*(f)\sigma^*(g)+(-1)^{\p{f}}\xi_\sigma^*(f)\sigma^*(g)$. In a similar way one can define gauge for gauge symmetries.

It is useful to represent gauge transformations in terms of vector fields defined on $M$. Given a vector field $Y$ of degree $-1$ it defines an infinitesimal symmetry $V=\commut{Q}{Y}$ of the Q-manifold. Because $\gh{V}=0$ it induces an infinitesimal transformation of the body of $M$, which is a gauge symmetry. Indeed, taking $\xi_\sigma=\sigma^* \circ Y$ one finds that the associated gauge symmetry is precisely the one determined by $V$.

It is also instructive to characterise solutions and gauge transformations in coordinate terms. For simplicity let us restrict ourselves to the situation where only coordinates $\cP_a,\phi^i,C^\alpha$ of degree $-1,0,1$ respectively are present. Generic $Q$ of degree $1$ then reads as:
\begin{equation}
Q=E_a(\phi)\dl{\cP_a}+C^\alpha R_\alpha^i(\phi)\dl{\phi^i}+\text{terms quadratic in $C,\cP$}\,.
\end{equation}
For a given map $\sigma$ one finds that $\sigma^*$ acts trivially on $C,\cP$ so
that the only nontrivial condition is $\sigma^*(Q \cP_a)=0$, giving $\sigma^*(E_a)=0$. In other words, solutions are  points where $E_a(\phi)$ vanish  so that  relations $E_a(\phi)=0$ are indeed interpreted as equations of motion.

Let us turn to gauge symmetries, a generic gauge parameter map $\xi_\sigma^*$ of degree $-1$ is determined by $\xi_\sigma^*(c^\alpha)=\epsilon^\alpha \in \fR$. The gauge variation of $z^i=\sigma^*\phi^i$ is then given by 
\begin{equation}
 \delta_{\xi_\sigma} z^i=\xi^*_\sigma( Q \phi^i)  = R_\alpha^i(\phi^i)\epsilon^\alpha
\end{equation}
In other words vector fields $R_\alpha^i(\phi)\dl{\phi^i}$ define a gauge distribution. It is involutive on the surface $E_a=0$ as a consequence of $Q^2=0$. 

A natural question is when two gauge systems should be considered equivalent. If the gauge invariance is not present it is natural to call systems equivalent if they have isomorphic spaces of solutions. Situation is slightly more tricky when gauge symmetry is present as some of the variables can be gauged away by passing to the space of gauge orbits and one should compare the quotients. A subtlety is that gauge symmetries may act in a nontrivial way as not all gauge transformations can be represented as shift transformations of the form $\delta z = \epsilon$, with $z$ being a global coordinate. These vague considerations lead to the following notion of equivalence for $Q$-manifolds:
\begin{definition}\cite{Grigoriev:2019ojp}
    1. A contractible $Q$-manifold is that of the form $(T[1]W,\mathrm{d_W})$ where $W$ is a graded vector space considered as a graded manifolds and $\mathrm{d_W}$ is the de Rham differential on $W$ seen as a $Q$-structure on $T[1]W$.
    
    2. A $Q$-manifold $(N,q)$ is called an equivalent reduction of $(M,Q)$ if $(M,Q)$ is a locally-trivial $Q$-bundle over $(N,q)$ which admits a global $Q$-section and whose fibers are contractible $Q$-manifolds. The equivalent reduction generates the equivalence relation for $Q$-manifolds.
\end{definition}
Recall that $Q$-bundle $\pi:(M,Q)\to (N,q)$ is a locally- trivial bundle of graded manifolds such that $\pi^*\circ  q=Q \circ \pi^*$ \cite{Kotov:2007nr}. A $Q$-bundle $\pi$ is called locally-trivial if it is locally isomorphic, as a $Q$-manifold,  to the product of the $(N,q)$ and a typical fiber $(F,q^\prime)$, i.e. $Q=q+q^\prime$ in the adapted coordinates.

The above notion of the equivalent reduction is a globalization of the  well-known, in the context of local BRST cohomology, concept of eliminating contractible pairs, see e.g. \cite{Barnich:1995ap,Brandt:1996mh}. 
In the field theory context the appropriate extension of the above equivalence was put forward in~\cite{Barnich:2004cr},  which in turn is a nonlagrangian version of the equivalence through the elimination of so-called generalised auxiliary fields~\cite{Henneaux:1990ua,Dresse:1990dj}.

\section{Local gauge theories as gauge PDEs}
\label{sec:gPDE}

\subsection{Gauge PDE}
\begin{definition}
Gauge PDE $(E,Q,T[1]X)$ is a $\fZ$-graded fiber bundle $\pi: E\to T[1]X$  whose total space is equipped with a homological vector field $Q$ of degree $1$ satisfying: $Q \circ \pi^*=\pi^* \circ \dx$
and $\dx$ is the de Rham differential on $X$ seen as a homological vector field on $T[1]X$. Moreover, $(E,Q,T[1]X)$ should be equivalent to a non-negatively graded gauge PDE.
\end{definition}
The last condition means that the surface  singled out by the equations of motion encoded in $Q$ (see below) is itself a smooth fiber bundle and that the role of all the negative degree variables
is to provide a resolution of this surface. In the usual jet-space BV formulation this condition corresponds to the acyclicity of the Koszul-Tate differential contained in $Q$. One often also imposes some extra conditions of a technical nature to exclude pathological examples.

We assume that $X$ is a real manifold (space-time manifold) while $E$ (and $T[1]X$) is a $\fZ$-graded one with the degree denoted by $\gh{\cdot}$. Note that gauge PDEs are special cases of $Q$-bundles~\cite{Kotov:2007nr}. However, gauge PDEs are necessarily $\fZ$-graded and are typically infinite-dimensional. Further details and references can be found in~\cite{Grigoriev:2019ojp} (see also \cite{Barnich:2010sw,Grigoriev:2010ic,Grigoriev:2012xg} for the earlier and less-geometrical version).

By definitionm, a field configuration $\sigma$ (i.e. a section $\sigma:T[1]X \to E$) is a solution if 
\begin{equation}
\label{aksz-eom}
    d_X \circ \sigma^*=\sigma^* \circ Q\,.
\end{equation}
The infinitesimal gauge transformation of the configuration $\sigma$ is given by
\begin{equation}
\label{aksz-gs}
    \delta\sigma^*=d_X \circ \xi_\sigma^*+\xi_\sigma^* \circ Q\,
\end{equation}
where $\xi_\sigma^*:\CC^\infty(E)\to \CC^\infty(T[1]X)$ is a gauge parameter map, which has a degree $-1$ and satisfies
\begin{equation}
\label{gp-cond}
    \xi_\sigma^*(fg)=(\xi_\sigma^*(f)\sigma^*(g)+(-1)^{\p{f}}\sigma^*(f)\xi_\sigma^*(g)\,.
\end{equation}
In a similar way one defines gauge for gauge symmetries. It can be also useful to parameterize gauge parameters in terms of ghost degree $-1$ vector fields $Y$ on 
$E$ as follows: $\xi_\sigma=\sigma^*\circ Y$, 
where $Y$ is understood as a derivation
of the algebra of functions on $E$. 
It is easy to check that condition~\eqref{gp-cond} is indeed satisfied. 

The notion of equivalence naturally extends from Q-manifolds to Q-bundles~\cite{Grigoriev:2019ojp}. An equivalent reduction of a gauge PDE $(E,Q,T[1]X)$ is a $Q$-subbundle $(G,Q_G,T[1]X)$ of $E$ such that the initial gauge PDE is a locally-trivial $Q$-bundle over $G$ with a contractible fiber. Moreover, the projection $E\to G$ is also a $Q$-map of $Q$-bundles over $T[1]X$. 

\subsection{Examples}

\subsubsection{PDE}
The first example is the usual PDE. Let $E_X \to X$ be an equation manifold seen as a bundle over the space-time $X$. Algebra of horizontal forms on $E_X$
can be identified with the algebra of functions on
$E$ which is $E_X$ pulled back to $T[1]X$ by the canonical projection $T[1]X\to X$. Under this identification the horizontal differential becomes a homological vector field on $E$, which by some abuse of conventions we also denote by $\dh$. In a suitable local coordinate system $x^a,\theta^a,\psi^A$, where $x^a, \theta^a$ are standard coordinates on $T[1]X$ pulled back to $E$, and $\psi^A$ originate from local coordinates on a typical fiber, one has:
\begin{equation}
    \dh=\theta^a D_a\,, \qquad D_a=\dl{x^a}-\Gamma_a^A(x,\psi)\dl{\psi^A}\,,
\end{equation}
where $D_a$ are components of the total derivative on the equation manifold. Note that this is a generic form of the $Q$-structure on $E$ if the typical fiber is a real manifold (so that $\gh{\psi^A}=0$). Using these coordinates it is easy to check that  solutions of $(E,\dh,T[1]X)$ are indeed the covariantly-constant (with respect to the Cartan distribution seen as an  Ehresmann connection) sections in agreement with the standard picture, see e.g.~\cite{Krasil?shchik-Lychagin-Vinogradov,Krasil'shchik:2010ij}. Note that a closely related representation of PDEs is known under the name of unfolded formulation~\cite{Vasiliev:1988xc,Vasiliev:2005zu}.

\subsubsection{AKSZ sigma model}
Let $X$ be a space-time manifold and $(F,q)$ a generic $Q$-manifold then take 
$(E,Q,T[1]X)$ to be a globally-trivial $Q$-bundle $(E,Q)=(F,q)\times (T[1]X,\dx)$ over $T[1]X$.  This construction is known as AKSZ sigma model and originates from~\cite{Alexandrov:1995kv}, where its Lagrangian version (in which case $F$ is equipped with a symplectic form of degree $\dim{X}-1$) was put forward as a tool to generate elegant BV formulations for topological models.

A typical example is to take $F=\algg[1]$ and $q=d_{CE}$, where $\algg[1]$ is a Lie algebra and $d_{CE}$ its Chevalley-Eilenberg differential. A section $\sigma$ is determined by the 1-form $A^\alpha_a(x)\theta^a=\sigma^*(c^\alpha)$, where $c^\alpha$ are linear coordinates on $\algg[1]$ and $x^a$ local coordinates on $X$. The gauge PDE equations of motion then imply that $\dx A +\half\commut{A}{A}=0$ and the  
gauge transformations are $\delta A=\dx\epsilon+\commut{A}{\epsilon}$.

\subsubsection{BV-formulation at the level of equations of motion}
\label{sec:EOM-BV}
The BV approach to local gauge theories was initially designed for Lagrangian systems in which case the system is defined in terms of a certain extension $S_{BV}$ of the classical action while the $Q$ structure emerges as a Hamiltonian vector field generated by $S_{BV}$ and the odd symplectic structure.  The version of BV for not necessarily Lagrangian systems deals with $Q$ structure from the very start while the symplectic structure and hence $S_{BV}$ is disregarded or not specified, see~\cite{Barnich:2004cr,Lyakhovich:2004xd} for further details.

In the setting of local gauge theories the BV formulation at the level of equations of motion consists of the following data: a fiber bundle $\cE \to X$, where $X$ is the space-time manifolds and the the typical fiber is a $\fZ$-graded manifold. The associated jet-bundle $J^\infty (\cE)$ is again a graded manifold and is equipped with a vertical evolutionary vector field $s$, called BRST differential,  satisfying $\gh{s}=1$ and $s^2=0$. Recall that evolutionary vector fields are those preserving the contact ideal. This condition can be also written as $\commut{D_a}{s}=0$, where $D_a=\dl{x^a}+\ldots$ denote total derivatives on $J^\infty (\cE)$, i.e. the horizontal lifts of the coordinate vector field $\dl{x^a}$ on $X$. 

This data encodes the information on equations of motion, gauge symmetries as well as gauge for gauge symmetries of the underlying system. Indeed, a section $\sigma:X\to \cE$ is called  a solution if its prolongation $\sigma_{\pr}:X\to J^\infty (\cE)$
satisfies $\sigma_{\pr}^*\circ s=0$. While gauge transformations are defined as $\delta \sigma_{\pr}^*=\xi_{\pr}^*\circ s$, where $\xi^*$ is a ghost degree $-1$ map $\cC^\infty(J^\infty (\cE)) \to \cC^\infty(X)$ satisfying $\xi^*(fg)=\xi^*(f)\sigma^*(g)+(-1)^{\p{f}}\xi^*(f)\sigma^*(g)$ for all $f,g\in \cC^\infty(J^\infty (\cE))$. In a similar way one defines gauge for gauge symmetries.

It is not difficult to construct a gauge PDE out of this data. Namely, take as $E \to T[1]X$ a pull back of $J^\infty(\cE)$ by the canonical projection $T[1]X\to X$. Functions on $E$ can be identified with the algebra of horizontal forms on $J^\infty(\cE)$. In particular, $E$ is equipped with the horizontal differential $\dh$ which can be written as $\theta^a D_a$, where $\theta^a$ are coordinates on the fibers of $T[1]X$. As a $Q$-structure of the total space one takes $Q=\dh+s$. It is easy to check that it satisfies $\pi^*\circ \dx=Q\circ \pi^*$, where $\dx$ is the de Rham differential on $X$ seen as a vector field on $T[1]X$. It was proved in~\cite{Barnich:2010sw} (see also~\cite{Barnich:2004cr} for the earlier but less general statement and \cite{Grigoriev:2019ojp} for a more geometrical exposition) that, at least locally, the gauge system encoded in $(E,Q,T[1]X)$ is an equivalent reformulation of the starting point BV-system.

\section{Presymplectic gauge PDEs}
\label{sec:presymp-gauge-PDE}

In this section we demonstrate  how generic Lagrangian local gauge field theories can be described by gauge PDEs equipped with compatible presymplectic structures. 

\subsection{Lagrangian BV system}
\label{sec:lagBVjet}

We begin by recalling the Lagrangian BV formulation of local gauge field theories in terms of jet-bundles, for a review see e.g.~\cite{Barnich:2000zw} and references therein.  
As in Section~\bref{sec:EOM-BV} we assume that $J^\infty(\cE)$ is equipped with a ghost grading and an evolutionary vector field $s$, satisfying $s^2=0,\gh{s}=1$. Now let, in addition, $\cE$ be equipped with a closed $(n+2)$-form of degree $-1$ and of the following structure:
\begin{equation}
\label{sympcE}
\omega^\cE=(dx)^n\omega_{AB}(x,\psi^A)d\psi^A d\psi^B\,,
\end{equation}
with $\omega_{AB}$ invertible. Here $x^a$ are local coordinates on $X$ pulled back to $\cE$ and $\psi^A$ the fiber coordinates. Although the requirement is formulated in coordinate terms it is actually invariant.\footnote{The alternative is to define a vertical 2-form with values in densities on $X$}

The above form defines an $(n,2)$ form $\st{n}{\omega}$ on the associated jet-bundle $J^\infty(\cE)$, which is just a pullback of $\omega^\cE$ to $J^\infty(\cE)$. It terms of standard local coordinates on $J^\infty(E)$ induced by $x^a,\psi^A$ it can be written as
\begin{equation}
    \st{n}{\omega}=(dx)^n\omega_{AB}(x,\psi^A)\dv\psi^A \dv\psi^B\,,
\end{equation}
where  $\dv$ is the vertical differential. Recall that on a jet-bundle de Rham differential decomposes as $d=\dh+\dv$, see e.g.~\cite{Anderson1991,Andersonbook,Dickey:1991xa} for more details on jet-bundles. If in addition $\st{n}{\omega}$ is $s$-invariant modulo a total derivative, i.e.
\begin{equation}
\label{descent}
L_s\st{n}{\omega}+\dh \st{n-1}{\omega}=0\,,
\end{equation}
for some $\st{n-1}{\omega} \in \bigwedge^{n-1,2}(J^\infty(\cE))$ we say that $(J^\infty(\cE),s,\st{n}{\omega})$ defines a local BV system (or local BV formulation). 

To see that the above data indeed encodes BV master-action let us restrict to local analysis so that $\dh$ is acyclic in horizontal form degree $<n$ and $d$  (as well as $\dv$) is acyclic in (vertical) form degree $>0$. 
Using $\dv \st{n}{\omega}=0$ and the fact that $\dv$ mod $\dh$ cohomology is trivial in this degree, one finds:
\begin{equation}
\label{descent2}
i_s \st{n}{\omega}=-\dv \st{n}{\hL}-\dh \st{n-1}\chi\,,
\end{equation}
for some  $\st{n}{\hL} \in \bigwedge^{n,0}(J^\infty(\cE))$ and some $\st{n-1}\chi \in \bigwedge^{n-1,1}(J^\infty(\cE))$. Applying $\dv$ to both sides one finds $L_s \st{n}{\omega}+\dh\dv \st{n-1}\chi=0$ so that $\st{n-1}{\omega}$ can be assumed $\dv$-closed (as we do in what follows).  Note that $\st{n}{\hL}$ is defined up to a total derivative, indeed $\st{n}{\hL}+\dh \st{n-1}{\epsilon}$
and $\st{n-1}\chi+\dv \st{n-1}{\epsilon}$ is again a solution to~\eqref{descent2} for any $\st{n-1}{\epsilon} \in \bigwedge^{n-1,0}(J^\infty(\cE))$.

Furthermore, considering $\dv(i_s i_s \st{n}{\omega})$ one finds (see Appendix for details):
\begin{equation}
\label{master}
 \half i_s i_s \st{n}{\omega}=\dh \st{n-1}{\hL}\,,
\end{equation}
for some $\st{n-1}{\hL}\in \bigwedge^{n-1,0}(J^\infty(\cE))$. Relations~\eqref{descent2} and \eqref{master}
are an equivalent reformulation of the standard axioms of local BV formulation, where $\cL_{BV}\equiv \st{n}{\hL}$ is the integrand of the BV master action and $s$ the associated BRST differential. In particular, \eqref{master} is equivalent to the BV master equation $\half\ab{L_{BV}}{L_{BV}}=\dh(\cdot)$, where the bracket is the odd Poisson bracket determined by $\st{n}{\omega}$. Note that the odd Poisson bracket on local $(n,0)$-forms is not unique and is defined up to a total derivative. See e.g. \cite{Barnich:1997ed} for a more detailed discussion of this point.

In what follows we also need the inverse statement:
\begin{prop}
\label{prop:inv}
Let $J^{\infty}(\cE)$
be equipped with a symplectic structure $\st{n}{\omega}$ determined (as explained above) by a nondegenerate $\omega^\cE$ and a vertical evolutionary vector field $s$, $\gh{s}=1$ (not necessarily nilpotent) such that equations~\eqref{descent2} and \eqref{master} are satisfied for some $\st{n}{\hL}$, $\st{n-1}{\chi}$, and $\st{n-1}{\hL}$.  It follows that $\st{n}{\hL},\st{n}{\omega}$ define a local BV system on $J^{\infty}(\cE)$.
\end{prop}
\begin{proof}
Let $s^\prime$ be 
an evolutionary vertical vector field determined by $i_{s^\prime}\st{n}{\omega}+\delta^{EL}\st{n}{\hL}=0$, where $\delta^{EL}$ is the Euler-Lagrange differential. This is the standard definition of a Hamiltonian vector field with Hamiltonian $\hL_{BV}$. It is the standard fact (see e.g.~\cite{Barnich:1997ed}) that nilpotency of $s^\prime$ is guaranteed if~\eqref{master}, with $s$ replaced by $s^\prime$, is satisfied for some $\st{n-1}{\hL}$. It is indeed satisfied as $i_s\omega - i_{s^\prime}\omega=\dh(\cdot)$ thanks to~\eqref{descent2} and the definition of $s^\prime$.
\end{proof}

\subsection{Descent-completed BV systems}
\label{sec:descent}
The above Lagrangian BV formulation can be extended in such a way that the symplectic structure $\st{n}{\omega}$ is completed by forms of lower horizontal degree to a cocycle of the total BRST differential $s+\dh$. To avoid discussion of possible obstructions we work locally. 

Using $\dv \st{n-1}{\omega}=0$ one finds that \eqref{descent} implies $\dh L_s \st{n-1}{\omega}=0$. Acyclicity of $\dh$ in horizontal form degree $<n$
then implies the next term in the descent equation:
\begin{equation}
    L_s \st{n-1}{\omega}+\dh \st{n-2}{\omega}=0\,.
\end{equation}
Repeating the same arguments as above gives
\begin{equation}
    i_s\st{n-1}{\omega}=-\dv \st{n-1}{\hL}-\dh \st{n-2}\chi
\end{equation}
for some $\st{n-1}{\hL} 
 \in \bigwedge^{n-1,0}(J^\infty(\cE))$ and $\st{n-2}\chi \in \bigwedge^{n-2,1}(J^\infty(\cE))$, which also implies that $\st{n-2}{\omega}$ can be assumed $\dv$-exact as we do in what follows. Continuing in the same way one finds $L_s \st{n-k}{\omega}+\dh \st{n-k-1}{\omega}=0$ for all $k>1$, which in turn implies the existence of $\st{n-k}{\hL} \in \bigwedge^{n-k,0}(J^\infty(\cE))$ and $ \st{n-k-1}{\chi} \in \bigwedge^{n-k-1,1}(J^\infty(\cE))$ such that
\begin{equation}
\begin{gathered}
\omega^{\tot}=\st{n}\omega+\st{n-1}\omega+\ldots+\st{1}\omega+\st{0}\omega\,, \qquad \hL^\tot=\st{n}{\hL}+
\st{n-1}{\hL}+\ldots+\st{1}{\hL}+\st{0}{\hL}\,,\\
\chi^{\tot}=\st{n}\chi+\st{n-1}\chi+\ldots+\st{1}\chi+\st{0}\chi\,,
\end{gathered}
\end{equation}
with $\st{n}\chi$ defined such that $\dv\st{n}\chi=\st{n}\omega$, satisfy
\begin{equation}
\label{descent-tot}
(\dh+L_s)\omega^{\tot}=0\,, \qquad  i_s \omega^{\tot}=-\dv \hL^\tot-\dh\chi^{\tot}\,, \qquad \omega^{\tot}=\dv \chi^{\tot}\,,
\end{equation}
and
\begin{equation}
\label{master-tot}
\half i_si_s \omega^{\tot}=\dh \hL^\tot\,.
\end{equation}
Note that $\omega^{\tot}, \hL^\tot$, and $\chi^{\tot}$ are inhomogeneous in the horizontal form degree. 

A BV jet-bundle equipped with 
$\omega^{\tot}, \hL^{\tot},\chi^{\tot}$ 
satisfying \eqref{descent-tot} and \eqref{master-tot} is refereed to in what follows as a descent-completed local BV system. Although, as we have just recalled, a usual local BV system can be extended, at least locally, to a descent completed one, such an extension is in general not guaranteed and/or is not unique. In this sense it contains extra data with respect to the local BV formulation reviewed in Section~\bref{sec:lagBVjet}. For instance, $\st{n-1}{\hL}$ and $\st{n-1}{\omega}$ correspond to, respectively, the BRST charge and the symplectic structure of the associated Hamiltonian formulation, see e.g.~\cite{Cattaneo:2012qu,Sharapov:2015wya,Sharapov:2016sgx}.   

In one or another version descent-completed BV systems are known in the literature.  In particular, the total BRST differential $s+\dh$ was actively employed from the early days of the local BRST cohomology, see e.g.~\cite{Stora:1984,Barnich:1995ap,Brandt:1996mh}.
Furthermore, the symplectic structure equivalent to $\omega^{\tot}$ and the equivalent of $\hL^{\tot}$ appeared already in~\cite{Grigoriev:2010ic,Grigoriev:2012xg} (see also a related approach of~\cite{Alkalaev:2013hta} and its developments~\cite{Grigoriev:2016wmk,Grigoriev:2020xec}) within a slightly different formalism.  
In this approach the essential part of $\hL^{\tot}$ is constructed as a descent completion of the Lagrangian. 
Detailed study of the descent completion of (pre)symplectic structures in the jet-bundle BV approach was put forward in~\cite{Sharapov:2015wya,Sharapov:2016sgx}. The notion of descent-completed~\footnote{called "Lax" in \cite{Simao:2021xgw}} BV was fully explicit in~\cite{Simao:2021xgw,Mnev:2019ejh} (see also~\cite{Cattaneo:2012qu}, where a related structure have appeared in a slightly different framework). 

\subsection{Presymplectic BV systems}
\label{sec:presympBV}
The usual Lagrangian BV formulation recalled in Section~\bref{sec:lagBVjet}
has a  straightforward generalization to the case where the symplectic structure $\omega^\cE$ introduced in~\eqref{sympcE}, is allowed to have a nonvanishing kernel on the vertical tangent vectors but its rank is required to be constant.

It turns out that almost all the considerations of Section~\bref{sec:lagBVjet} remain true. In particular, $\omega^\cE$  determines $\dv$-closed $\st{n}{\omega} \in \bigwedge^{n,2}(J^\infty(\cE))$. If BRST differential $s$ preserves $\st{n}{\omega}$ modulo boundary terms, i.e. equation~\eqref{descent} is satisfied, one finds $\st{n}{\hL} \in \bigwedge^{n,0}(J^\infty(\cE))$ satisfying \eqref{descent2} and \eqref{master} for some $\st{n-1}{\chi}$ and $\st{n-1}{\hL}$.

Let us restrict to local analysis and take a quotient of $\cE$ modulo the distribution of vertical vector fields that are in the kernel of $\omega^\cE$. This can be performed fiberwise, giving a new bundle $\cG$ together with a projection $p:\cE \to \cG$.
$\omega^\cE$ induces a nondegenerate $n+2$ form $\omega^\cG$ on $\cG$, which has the same structure~\eqref{sympcE}.
Recall that locally a symplectic quotient can be realised as a submanifold and in the case at hand as a subbundle (because the distribution is vertical). To simplify the analysis we realise $\cG$ as a subbundle of $\cE$ and hence $J^{\infty}(\cG)$ is also a  subbundle of $J^{\infty}(\cE)$ in a natural way. Pulling back equation~\eqref{descent2} to $J^{\infty}(\cG)$ one finds that it remains true with $s$ replaced by its projection to $J^{\infty}(\cG)$. At the same time it is easy to see that the projection of $s$ to  $J^{\infty}(\cG)$ is evolutionary. It then follows that we arrived at the local BV system on $J^{\infty}(\cG)$ with the symplectic structure $\omega^{\cG}$ and the BV master action being a pullback of $\st{n}{\hL}$ to $J^{\infty}(\cG)$.

The above construction can be also understood in field-theoretical terms. Suppose that $R$ is a vertical  evolutionary vector field from the kernel of $\st{n}{\omega}$, i.e. $i_R \st{n}{\omega}=0$. It follows it defines a gauge symmetry of the BV Lagrangian, Indeed, equation~\eqref{descent2} then implies
\begin{equation}
    L_R \st{n}{\hL}+(-1)^{\p{R}}\dh i_R \st{n-1}{\chi}=0\,.
\end{equation}
From this perspective the passage to $\cG$ can be seen as gauge-fixing these symmetries. Note that one can choose $R$ to be a prolongation of a kernel vertical vector field on $\cE$ in which case the respective gauge transformations are algebraic.

\subsection{Presymplectic gauge PDEs}\label{sec:pPDE}

Let us first introduce vertical forms on $(E,Q,T[1]X)$ as the quotient of all local forms by the ideal $\cI$ generated by elements of the form $\pi^*\alpha$, where $\alpha$ is a form of positive degree on $T[1]X$ and $\pi$ is the projection $\pi:E\to T[1]X$. In coordinate terms this means that forms proportional to $dx^a$ or $d\theta^a$ are to be considered equivalent to zero. An equivalence class of a representative $\rho\in\bigwedge^{\bullet}(E)$ is denoted by $[\rho]$. The de Rham differential is clearly well defined on the quotient as well as the Lie derivative along the vector fields that project to the base, e.g. vertical vector fields.

The following definition is taken from~\cite{Dneprov:2022jyn} though less general versions were  in~\cite{Grigoriev:2016wmk,Grigoriev:2020xec} while the examples appeared already in~\cite{Alkalaev:2013hta}:
\begin{definition}\label{def:presymppde}
A compatible presymplectic structure on gauge PDE $(E,Q,T[1]X)$ is a vertical two form $[\omega]$ of degree $n-1$, $n=\dim{X}$ satisfying  $[d\omega]=0$, $[L_Q\omega]=0$. A gauge PDE equipped with a compatible presymplectic structure is called presymplectic gauge PDE.
\end{definition}
In what follows we find it useful to work in terms of a fixed representative $\omega$ of $[\omega]$. 

It turns out that under certain regularity conditions a presymplectic gauge PDE defines a local BV system (and in fact a descent-completed BV system as well). To construct it consider a super-jet bundle $SJ^\infty(E)$ of $\pi:E\to T[1]X$, i.e. a bundle of jets of super-sections. It is again a graded manifold which is a bundle over $E$ and, of course, over $T[1]X$. We refer to~\cite{Khudaverdian:2001qe} for general discussion of super-jet bundles and to~\cite{Grigoriev:2019ojp} for precisely this setup.

Let $V$ be a vector field on $E$. Then its prolongation is defined as a unique vector field $\bar V$ on $SJ^\infty(E)$ that projects to $V$ by $\Pi:SJ^\infty(E) \to E$ and whose vertical component is evolutionary. Applying this to $Q$ gives its prolongation $\bar Q$ which we decompose as $\bar Q=D+s$, where $D$ is horizontal and $s$ vertical. Because prolongation is a canonical operation one has $\gh{\bar Q}=1$ and $(\bar Q)^2=0$. Furthermore, because $\bar Q$ projects to $\dx$ under $\pi\circ \Pi$, $SJ^\infty(E)$ equipped with $\bar Q$ and seen as a bundle over $(T[1]X,\dx)$ is again a gauge PDE.  Moreover, it follows $D$ is a horizontal lift of $\dx$, and hence, $D=\theta^a D_a$, where $D_a$ is a total derivative (horizontal lift of $\dl{x^a}$). If $u^A$ are fiber coordinates of $E$ pulled back to $SJ^\infty(E)$, then $su^A=Qu^A-\theta^a D_au^A$. The action of $s$ on the remaining coordinates is determined by $\commut{D_a}{s}=0=\commut{D^\theta_a}{s}$, where $D^\theta_a$ is a horizontal lift of $\dl{\theta^a}$ and the commutator always denote the graded commutator. 

In what follows we need to employ a useful coordinate system on $SJ^\infty(E)$. Given an adapted coordinate system $x^a,\theta^a,u^A$ the associated coordinates are $x^a,\theta^a$, $\psi^A_{a_1\ldots a_l|b_l\ldots b_k}$ which are defined as follows: $\psi^A_{a_1\ldots a_l|b_1\ldots b_k}=D_{a_1}\ldots D_{a_l}\psi^A_{|b_1\ldots b_k}$, $D^\theta_a=\dl{\theta^a}$,  and $\psi^A_{|b_1\ldots b_k}$ coincides with $D^\theta_{b_1}\ldots D^\theta_{b_k} u^A$ at $\theta^a=0$. In particular, $u^A=\psi^A+\theta^a \psi^A_{|a}+\half\theta^a \theta^b\psi^A_{|ba}+\ldots$, for further details see e.g.~\cite{Grigoriev:2019ojp}. Note that in this coordinate system $\commut{D^\theta_0}{s}=\commut{\dl{\theta^a}}{s}=0$, i.e. components of $s$ are $\theta$-independent while $D=\theta^aD_a$. 

The $2$-form $\omega$ also admits a canonical lift $\bar\omega$ to $SJ^\infty(E)$, namely $\bar\omega=\Pi^*(\omega)$. In terms of $\psi^A_{\cdot|\cdot}$-coordinates, $\bar\omega$
is given by $\omega_{AB}(\psi^A(\theta))d\psi^A(\theta)d\psi^B(\theta)$ and can be decomposed into its $\theta$-homogeneous components:
\begin{equation}
\bar\omega=\st{n}{\bar\omega}{}+\st{n-1}{\bar\omega}{}+\ldots+\st{0}{\bar\omega}{}\,.
\end{equation}

Just like presymplectic structure $\omega$ on $E$ we consider $\bar\omega$ on $SJ^\infty(E)$ as a representative of an equivalence class modulo the ideal generated by $dx^a,d\theta^a$ pulled back from $T[1]X$ to  $SJ^\infty(E)$. By some abuse of notations we keep denoting the ideal by $\cI$.
However, in contrast to $E\to T[1]X$, super-jet bundle
$SJ^\infty(E)$ comes equipped with a canonical Cartan distribution, i.e. the horizontal distribution generated by $D_a,D^{\theta}_a$, so that one can pick a vertical component $\bar\omega^v$ as a unique representative of the equivalence class. Note that $\bar\omega^v$ only depends on the equivalence class $[\omega]$ of $\omega$.

A vertical evolutionary vector field on $SJ^\infty(E)$  preserves the subspaces of horizontal and vertical forms.  The same applies to $D=\theta^a D_a$. Because prolongation is canonical $L_{Q} \omega \in \cI$ implies $L_{\bar Q}\bar\omega \in \cI$. Taking a purely vertical part of both sides gives $L_{\bar Q}\bar\omega^v=0$.  Decomposing $\bar\omega^v$ into the $\theta$-homogeneous  components we get the usual descent equation:
\begin{equation}
L_{s}\st{n}{\bar\omega}{}^{v}+L_D \st{n-1}{\bar\omega}{}^{v}=0\,, \quad L_{s}\st{n-1}{\bar\omega}{}^{v}+L_D \st{n-2}{\bar\omega}{}^{v}=0\,, \quad \ldots\,,
\end{equation}
which says that $\st{k}{\bar\omega}{}^{v}$ is $s$-invariant modulo a total derivative. Indeed, the algebra of vertical forms on $SJ^\infty(E)$ can be identified with the algebra of all local forms on $SJ^\infty(E)_0$ which is, by definition, $SJ^\infty(E)$ restricted to the zero section of $T[1]X$ (i.e. in standard coordinates $SJ^\infty(E)_0$ is a submanifold in $SJ^\infty(E)$ singled out by $\theta^a=0$). Moreover, $SJ^\infty(E)_0$ can be seen as a jet-bundle $J^\infty(\bar E)$ of another bundle $\bar E \to X$. Under this identification $\theta^a$ goes to $dx^a$ while $L_D$ goes to $\dh$ so that the above equations coincide with the first equation in~\eqref{descent-tot}. In the coordinate terms, this identification simply amounts to disregarding $dx^a,d\theta$ (which anyway do not enter vertical forms) and then replacing $\theta^a$ with $dx^a$.

In this way we arrive at the presymplectic BV system whose defining data consist of a vertical evolutionary vector field $s$ and a $2$-form $\st{n}{\bar\omega}$ of $\theta$-degree $n$, ghost degree $n-1$, and  satisfying $L_s \st{n}{\bar\omega}+L_D \st{n-1}{\bar\omega}=0$. Indeed, the analysis of Section~\bref{sec:lagBVjet} applies to a degenerate  $\st{n}{\bar\omega}$ as well, and hence there exists a $\theta$-homogeneity $n$ local function $\cL_{BV}$ satisfying:
\begin{equation}
\label{BV-presymp}
\dv \cL_{BV}+i_{s}\st{n}{\bar\omega}{}^v=L_D(\cdot)\,, \qquad  i_s i_s \st{n}{\bar\omega}{}^v=D(\cdot)\,.
\end{equation}
Note that $\cL_{BV}$ is not unique as it is defined modulo $D$-exact terms. Note also that $s,\bar\omega{}^v$ contains extra data with respect to just BV-symplectic form and $s$. Namely, the descent completion of $\st{n}{\bar\omega}{}^v$. 

To arrive at  a genuine local BV system it is again convenient to identify $SJ^\infty(E)_0$ as $J^\infty(\bar E)$. Considered as an $(n,2)$ form on $J^\infty(\bar E)$, form $\st{n}{\bar\omega}{}^v$ is a pull-back from $\bar E$ (does not involve derivatives) and hence we are in the setting of Section~\bref{sec:presympBV}.
It follows that if $\st{n}{\bar\omega}{}^v$ is regular this defines, at least locally, a local BV system on the associated symplectic quotient.

Of course, this BV system captures only part of the initial information. For instance, descent completion of $\st{n}{\bar\omega}{}^v$ is not encoded in this BV system. In this respect, the presymplectic gauge PDE provides a more detailed description. Furthermore,  the data of presymplectic gauge PDE, as defined above, determines $\cL_{BV}$ only modulo $\dh$-exact terms (boundary terms). As we discuss in the next Section, to fix this ambiguity one needs to choose a presymplectic potential $\chi \in \bigwedge^1(E)$ such that $\omega -d\chi\in \cI$. 

It is important to stress that the local BV system encoded in $(E,Q,T[1]X,\omega)$ again defines a gauge PDE. Indeed, forgetting the BV symplectic structure and applying the procedure explained in Section~\bref{sec:EOM-BV} one arrives at the gauge PDE.  This gauge PDE is in general not equivalent to the initial $(E,Q,T[1]X,\omega)$ with $\omega$ forgotten but is weaker (has more inequivalent solutions).  

It is also important to note that a descent-completed BV system discussed in Section~\bref{sec:descent}
gives examples of presymplectic gauge PDE, with $Q$ and $\omega$ being respectively $\dh+s$ and $\omega^{tot}$.  In this way we conclude that any descent-completed BV system can be seen as a presymplectic gauge PDE of a special type. As we are going to see in the next section, the inverse statement is also true in the sense that, at lest locally, a presymplectic gauge PDE  determines a descent-completed BV formulation, provided one has fixed a symplectic potential $\chi$ such that $\omega-d\chi \in\cI$.

\subsection{AKSZ-like representation of the Lagrangian}
\label{sec:aksz-like}

Let us make an explicit contact with the AKSZ-like representation of the Lagrangian $\hL_{BV}$ or, more precisely, the respective BV master-action. Given a presymplectic gauge PDE
$(E,T[1]X,Q,[\omega])$ let us pick a particular representative $\omega$ of $[\omega]$ and a 1-form $\chi$ (symplectic potential) such that $\omega=d\chi$. In addition, we impose the following  extra technical condition $i_QL_Q\omega=\cI$. It follows, there exists function $\hL$ such that $i_Q\omega+d \hL\in \cI$ so that the full set of relations read as:
\begin{equation}
Q^2=0\,, \qquad \omega=d\chi\,, \qquad i_Q\omega+d\hL\in \cI\,, \qquad 
i_QL_Q\omega \in \cI\,.
\end{equation}
In what follows we refer to the above system as a descent-completed presymplectic gauge PDE and denote it by $(E,Q,T[1]X,\chi,\cL)$. Note that as shown in the Appendix these relations also imply $i_Qi_Q \omega=0$ and $Q\hL=0$.

Let us pick a local trivialization $x^a,\theta^a,\psi^A$ of $E$ and decompose  $Q$ as
\begin{equation}
Q=\dx+q\,, \qquad \dx=\theta^a\dl{x^a}\,,\quad q=q^A(\psi,x,\theta)\dl{\psi^A}\,.
\end{equation}
Of course, the decomposition depends on the choice of trivialization and, in general, is defined only locally. Consider the prolongation of $Q$ to $SJ^\infty(E)$ for $\dx$ and $q$ separately. In the local coordinates $\psi^A_{\cdot|\cdot}$ employed in Section~\bref{sec:pPDE} the prolongation of $\dx$ reads as:
\begin{equation}
(\dx)^{\pr}=D+d^F\,, \qquad D=\theta^a D_a\,,
\end{equation}
where $d^F$ is a vertical part of $(\dx)^\pr$ defined in Section~\bref{sec:pPDE}.

Let $\Pi$ denote a canonical projection of $SJ^\infty(E)$, considered as a bundle over $E$, i.e. $\Pi:SJ^\infty(E)\to E$. Differential forms and vector fields admit a natural lift from $E$ to $SJ^\infty(E)$: for differential forms the lift is just a pullback by $\Pi$ while for vector fields the lift is given by prolongation. Because by definition prolongation of a vector field $V$ on $E$ is a vector field $V^\pr$ such that it projects to $V$ by $\Pi$ and whose vertical part is evolutionary, one has:
\begin{equation}
i_{V^\pr}(\Pi^*(\alpha))=\Pi^*(i_V\alpha)\,, \qquad \alpha\in \bigwedge{}^l(E)    \,.
\end{equation}
In particular:
\begin{equation}
   i_D \bar\omega + i_{d^F} \bar\omega = \Pi^*(i_{\dx}\omega)\,,\qquad \bar \omega=\Pi^*\omega\,,
\end{equation}
\begin{equation}
   i_{\bar q} \bar\omega + 
   \Pi^*(i_{\dx}\omega)+d\bar \hL \in \cI\,,\qquad \bar q=q^{\pr}\,, \quad \bar \hL=\Pi^* \hL\,.
\end{equation}
Furthermore, $\omega=d\chi$ implies $\bar\omega=d\bar\chi$ so that  
\begin{equation}
    i_D\bar\omega=i_D (d\bar\chi)=
    L_D \bar\chi+d(i_D\bar\chi)\,.
\end{equation}
Combining the above relations one finds:
\begin{equation}
    i_{d^F}\bar\omega+ i_{\bar q} \bar\omega = - d(\bar \hL+i_D\bar\chi)-L_D \bar\chi+\cI\,.
\end{equation}
Taking the vertical part and using $\bar q+d^F=s$ one gets
\begin{equation}
\label{presymp-tot}
    i_{s}{\bar\omega}{}^v= -\dv( { \bar \hL}+i_{D}{\bar\chi})-L_D {\bar\chi}{}^v\,,
\end{equation}
which is nothing but the second   relation~\eqref{descent-tot} with 
$\hL^\tot={ \bar \hL}+i_{D}{\bar\chi}$,
$\omega^\tot={\bar\omega}{}^v$, and
$\chi^\tot={\bar\chi}{}^v$.
The remaining relations in \eqref{descent-tot} are consequences of the above one and $\bar\omega{}^v=\dv \bar \chi{}^v$.

Let us check that~\eqref{master-tot}
is also satisfied.  As demonstrated in the Appendix
$i_Qi_Q\omega=0$ as well as $Q\hL=0$ so that $i_{\bar Q} i_{\bar Q} \bar\omega=\Pi^*(i_{Q} i_{Q} \omega)=0$ and hence:
\begin{equation}
i_s i_s \bar\omega=i_{\bar Q-D} i_{\bar Q-D} \bar\omega=-2i_{\bar Q} i_{D}\bar\omega+ i_D i_D\bar\omega\,.
\end{equation}
Using $\omega=d\chi$ one finds
\begin{equation}
i_D i_D\bar\omega=i_D i_D d\bar\chi=
 i_DL_D \bar\chi+i_D di_D\bar\chi=L_D(i_D\bar\chi)+D(i_D\bar\chi)=2D(i_D\bar\chi)\,,   
\end{equation}
where we made use of $D^2=0$. Furthermore, applying $i_Q$ to both sides of $i_Q\omega+d\hL=\alpha$ one finds that $i_Q\alpha=0$ thanks to $i_Qi_Q\omega=0=Q\hL$. Because $\alpha$ is a 1-form from $\cI$ one has $i_q \alpha=0$ so that $i_{\dx}\alpha=0$ and  
\begin{equation}
i_{\bar Q} i_{D} \bar\omega=-i_D (d \bar \hL+\Pi^*(\alpha))=-D\bar \hL- \Pi^*(i_{\dx}\alpha)=-D\bar \hL\,,
\end{equation}
where we made use of $i_D \Pi^* (\alpha)=\Pi^*(i_{\dx}\alpha)$. All in all this gives:
\begin{equation}
\label{presymp-master-tot}
\half i_s i_s \bar\omega=D(i_D \bar \chi+\bar \hL)\,.
\end{equation}
We conclude that all the relations~\eqref{descent-tot} and \eqref{master-tot}  are satisfied and hence:
\begin{prop}
Decent completed gauge PDE $(E,Q,T[1]X,\chi,\hL)$ defines a structure of the descent-completed BV formulation on $SJ^\infty(E)_0$ with possibly degenerate (in the sense of Section~\bref{sec:presympBV}) presymplectic structure.
\end{prop}
Here we do not address the issue of how to perform a symplectic reduction in this setup.

Finally, extracting the $\theta$-homogeneity $n$ component of \eqref{presymp-tot}  gives:
\begin{equation}
    i_{s}\st{n}{\bar\omega}{}^v= -\dv( \st{n}{ \bar \hL}+i_{D}\st{n-1}{\bar\chi})-L_D \st{n-1}{\bar\chi}{}^v\,.
\end{equation}
It follows $\st{n}{\bar \hL}+i_{D}\st{n-1}{\bar\chi}$ can be taken as a BV Lagrangian $\cL_{BV}$
on $SJ^{\infty}(E)$. The respective BV action can be written as:
\begin{equation}
\label{aksz-action}
S_{BV}[\hat\sigma]=\int_{T[1]X} \hat\sigma^*(\chi)(d_X) +\hat\sigma^*(\hL) \,, 
\end{equation}
where $\hat\sigma$ is a supersection of $E$.  

Of course, $S_{BV}$ does not depend on the fields parameterizing the kernel of the presymplectic structure and interpreting it as a standard BV action requires passing to the symplectic quotient or imposing additional gauge conditions which set these fields to particular configurations. With these additional gauge conditions the BV action can be used in the path integral just like the conventional BV-AKSZ action. 

The above BV action is a generalization of that for AKSZ sigma-models~\cite{Alexandrov:1995kv}, see also~\cite{Alkalaev:2013hta,Grigoriev:2020xec} for a version applicable to diffeomorphism-invariant systems and~\cite{Grigoriev:2016wmk,Grigoriev:2021wgw} for PDEs equipped with a symplectic current.\footnote{In fact the BV action of Lagrangian formulation proposed in~\cite{Grigoriev:2010ic,Grigoriev:2012xg} also has the same structure.} Note also that in the case $n=1$ and an invertible $\omega$ this construction was already in~\cite{Grigoriev:1999qz}.

To conclude, presymplectic gauge PDEs give a rather flexible  way to encode Lagrangian BV formulation of local gauge systems without employing jet-bundles. It can also be thought of as a generalisation of the BV-AKSZ formalism to the case of not necessarily topological and not necessarily diffeomorphims-invariant systems.

\section{Weak presymplectic gauge PDEs}
\label{sec:weak}
\subsection{Weak presymplectic $Q$-manifolds}

Despite presymplectic gauge PDEs being a rather flexible and invariant way to describe Lagrangian gauge systems, they are not very concise and explicit. For instance, for a nonilinear PDE it is often impossible to find an explicit expression for the underlying $Q$-structure and one only has the existence statement~\cite{Barnich:2010sw}.  An alternative would be to work in terms of the usual Lagrangian BV formalism, where the BV Lagrangian and the symplectic structure come from finite jets and can often be given explicitly, but in this way one is stuck to jet-bundles and hence looses flexibility.

It turns out that the data of a presymplectic gauge PDE can be encoded in a more flexible and often finite-dimensional bundle over $T[1]X$, where condition $Q^2=0$ is relaxed in a controlled way. Let us first illustrate the idea in the case of a toy model of "weak presymplectic gauge PDE" with $\dim{X}=0$.  
\begin{prop}
Suppose we are given with a graded manifold $E$ equipped with a vector field $Q$, $\gh{Q}=1$ and a closed 2-form $\omega$, $\gh{\omega}=l$ satisfying
\begin{equation}
d\omega=0\,, \qquad L_Q \omega=0\,, \qquad i_Q i_Q \omega=0\,,
\end{equation}
for some function $\hL$. If in addition $\omega$ is regular, its symplectic quotient, possibly locally defined, is naturally a symplectic $Q$-manifold. 
\end{prop}
\begin{proof}
Let $\cK \subset TE$ be a kernel distribution of $\omega$. For any $K\in \cK$ one has $\commut{K}{Q}\in \cK$ thanks to  $L_Q\omega=0$. For any function $f$ satisfying $\cK f=0$  one has $\cK Qf=0$ so that $Q$ induces a vector field on the symplectic quotient.  Let $p:E\to E^\prime$ denote a projection to the symplectic
quotient and $Q^\prime$ and $\omega^\prime$ denote the respective structures induced on the symplectic quotient. It follows 
$\omega=p^*\omega^\prime$ and $p^*\circ Q^\prime= Q\circ p^*$, i.e. $Q$ projects to $Q^\prime$. It then follows that $i_{Q^\prime} i_{Q^\prime}\omega^\prime=0$, $L_{Q^\prime}\omega^\prime=0$, and also $i_{Q^\prime}\omega^\prime+d \hL^\prime=0$ for some $\hL^\prime$. Because $\omega^\prime$ is symplectic, one concludes that $(Q^\prime)^2=0$. In particular, Hamiltonian $\hL^\prime$ of $Q^\prime$ satisfies master equation $\pb{\hL^\prime}{\hL^\prime}=0$, where the bracket denotes the (graded) Poisson structure determined by $\omega^\prime$.
\end{proof}

\subsection{Weak presymplectic gauge PDEs}
The notion and the statement can be generalized to the case where
$\dim{X}>0$ by replacing $Q^2=0$ with in general weaker condition $i_Qi_Q\omega=0$ and keeping the remaining conditions intact.
This gives a more flexible version of the presymplectic BV-AKSZ formalism developed in the previous section. 
More precisely, a weak presymplectic gauge PDE is a $\fZ$-graded fiber bundle over $T[1]X$, equipped with a vector field $Q$ satisfying $\gh{Q}=1$, $Q\circ \pi^*=\pi^* \circ \dx $  and a 2-form $\omega$ of degree $n-1$, $n=\dim{X}$ satisfying:
\begin{equation}
d\omega=0,\qquad  L_Q\omega \in \cI\,,\qquad
i_Qi_Q\omega=0\,,\qquad i_Q L_Q \omega\in \cI\,.
\end{equation}
In addition $\omega$ is assumed regular (in the same sense as in Section~\bref{sec:pPDE}, i.e. that the induced $\st{n}{\bar\omega}{}^v$ is regular on $SJ^\infty(E)$). This definition can be made even more flexible but it is enough to illustrate the idea. Note that the above relations imply $i_{Q^2}\omega \in \cI$.

In order to see that such system still defines a local BV system, let us observe that in the analysis 
of Section~\bref{sec:aksz-like}
we only employed $Q^2=0$ to derive $i_Q i_Q\omega=0$ and $Q\hL=0$. 
Now $i_Q i_Q\omega=0$ by assumption while $Q\hL=0$ by the same reasoning as in the Appendix. Picking $\chi$ and $\hL$
such that $\omega=d\chi$ and $i_Q\omega+d\hL\in \cI$ and repeating the considerations of Section~\bref{sec:aksz-like} in this more general situation we arrive at equations \eqref{presymp-tot} and \eqref{presymp-master-tot} and hence to a descent-completed presymplectic BV system. Because we assumed $\st{n}{\bar\omega}{}^v$ regular this in turn defines a local BV system.

A remarkable feature of such an extension of the presymplectic BV-AKSZ formalism is that it is often enough to consider finite-dimensional $E$ and there is no need to work with infinite-dimensional manifolds.  It does not mean that one should necessarily relax $Q^2=0$ to achieve that, as in some cases, including Maxwell field and Einstein gravity, there exist finite-dimensional presymplectic gauge PDEs encoding the full-scale BV formulation~\cite{Grigoriev:2020xec,Dneprov:2022jyn} (without BV extension similar examples were already in~\cite{Alkalaev:2013hta}).

\subsection{Example}
\label{sec:YM}
It is not difficult to give a nontrivial example of the weak presymplectic AKSZ system.  Consider the following weak presymplectic Q-manifold: coordinates are $x^a,\theta^a,C,F^{ab}$, where $\theta^a,C$ are of degree 1 and coordinates $C,F^{ab}$ take values in a reductive real Lie algebra. The weak $Q$-structure is determined by:
\begin{equation}
Qx^a=\theta^a\,, \quad Q C=-\half\commut{C}{C}+\half \theta^a\theta^b F_{ab}\,, \quad QF_{ab}=\commut{F_{ab}}{C}   \,. 
\end{equation}
Note that $Q$ is not nilpotent: $Q^2 F_{ab}$ is proportional to $\commut{F_{ab}}{\theta^c\theta^dF_{cd}}$. 
Of course, this $Q$-manifold is just a truncation of the minimal BRST complex of Yang-Mills theory, see e.g.~\cite{Brandt:1996mh}. 

As a presymplectic potential we take~\cite{Alkalaev:2013hta}:
\begin{equation}
\chi=\tr (\theta^{(n-2) F^{ab} }_{ab} dC)\,, \qquad
(\theta)^{(n-k)}_{a_1\ldots a_k}=\frac{1}{(n-k)!}\epsilon_{a_1\ldots a_k b_1\ldots b_{n-k}}\theta^{b_1}\ldots \theta^{b_{n-k}}\,.
\end{equation}
where the indexes of $F$ are raised and lowered by the constant Minkowski metric $\eta_{ab}$. It is easy to check that $i_Qi_Q\omega=0$, where $\omega=d\chi$. The "covariant Hamiltonian" $\hL$ determined through
$i_Q\omega+d\cL\in \cI$ reads as
\begin{equation}    
\hL=\half\tr(F^{ab}\commut{C}{C})\theta^{(n-2)}_{ab}-\half\tr(F_{ab}F^{ab})\theta^{(n)}\,.
\end{equation}

The BV-AKSZ action~\eqref{aksz-action} takes the form of the usual first order formulation of YM action:
\begin{equation}
S_{BV}=\int d^n x (F^{ab}(\d_a A_b-\d_b A_a+\commut{A_a}{A_b})-\frac{1}{2}F_{ab}F^{ab})+\ldots\,,  \end{equation}
where we only listed explicitly the terms not involving variable of nonvanishing ghost degree, $A_a(x)\theta^a=\sigma^*(C)$, and by some abuse of notations $F^{ab}(x)=\sigma^*(F^{ab})$.  Passing to the simplectic quotient gives the respective BV field-antifield space. The last step is completely analogous to the one in the case of the Maxwell field considered in~\cite{Dneprov:2022jyn}, to which we refer for further details.  

\section{Induced gauge PDEs
on submanifolds and boundaries}\label{sect:boundary}

As we have seen, presymplectic gauge PDEs can be considered as a far-going generalization of AKSZ sigma models. It is well-known that for a given AKSZ model one can replace the space-time manifold with its submanifold, or add extra dimension(s) without spoiling consistency. A typical example is given by $n=1$ AKSZ sigma model whose restriction to a point (i.e. just the target space) turns out to be the respective  BFV phase space~\cite{Grigoriev:1999qz}. More generally,  the Hamiltonian BFV description of generic AKSZ models can be obtained by restricting the model to a codimension 1 space-like submanifold~\cite{Barnich:2003wj}. This extends to a rather general gauge theories by representing them as AKSZ sigma models using so-called parent formulation approach~\cite{Grigoriev:2010ic,Grigoriev:2012xg}. In fact the Lagrangian parent formulation of these Refs. is a special class of presymplectic gauge PDEs. In the non-Lagrangian setup, gauge PDEs give a natural tool to study induced theories on submanifolds and (asymptotic) boundaries, as was demonstrated in~\cite{Bekaert:2012vt,Bekaert:2013zya,Chekmenev:2015kzf} (see also~\cite{Bekaert:2017bpy,Grigoriev:2018wrx} for further applications) in the context of AdS/CFT correspondence for higher-spin gauge theories. Let us also mention an alternative approach to BV systems on manifolds with boundaries~\cite{Cattaneo:2012qu,Cattaneo:2015vsa,Canepa:2020ujx,Simao:2021xgw} (see also Refs. therein), which operates in the usual BV setup and does not explicitly employ (presymplectic) gauge PDEs.

The simplest situation is where $(E,Q,T[1]X)$ is a gauge PDE and we specify a submanifold $\Sigma \subset X$  (e.g. $\Sigma$ can be a boundary of $X$) of its  space-time manifold $X$. There is a naturally defined induced gauge PDE structure on $T[1]\Sigma$. Indeed, $T[1]\Sigma$ can be seen as a $Q$-submanifold of $T[1]X$, i.e. $\dx$ is tangent to $T[1]\Sigma$. Introducing local coordinates $x^i,z^\alpha$ adapted to the embedding so that $\Sigma$ is singled out by $z^\alpha=0$, submanifold $T[1]\Sigma \subset T[1]X$ is singled out by $z^\alpha=0,\dx z^\alpha=0$.

Let $E_\Sigma \to T[1]\Sigma$ be a restriction of $E$ to 
$T[1]\Sigma \subset T[1]X$. It is easy to see that $Q$ is tangent to $E_\Sigma$, seen as a submanifold of $E$, and hence defines a $Q$-structure $Q_\Sigma$ on $E_\Sigma$, which projects to $\mathrm{d_\Sigma}$ and hence defines a new gauge PDE $(E_\Sigma, Q_\Sigma,T[1]\Sigma)$. In the case where $\Sigma$ is a boundary  of $X$ this can be regarded as a theory of boundary values of the fields described by $(E,Q,T[1]X)$. This is precisely the procedure applied in~\cite{Bekaert:2012vt,Bekaert:2013zya} to study boundary values of higher-spin fields in the AdS/CFT setup. 

Let us now turn to presymplectic gauge PDEs. The pullback $\omega_\Sigma$ of presymplectic structure $\omega$ to $E_\Sigma$ gives a presymplectic structure in $(E_\Sigma,Q_\Sigma,T[1]\Sigma)$. However, the ghost degree of $\omega_\Sigma$ is $\dim{X}-1$  and not $\dim{\Sigma}-1$ so that it is not a usual presymplectic gauge PDE. For $\Sigma$ of codim $1$ we are in fact dealing with 
 a version of Hamiltonian BFV system, as was initially observed in~\cite{Grigoriev:1999qz,Barnich:2003wj} in the AKSZ setup. 

 We refrain from giving here a general framework for presymplectic gauge PDEs induced on submanifolds and boundaries but  consider \Ruth{instead} a weak presymplectic gauge PDE, describing Yang-Mills system and presented in Section~\bref{sec:YM}. For simplicity we assume that $X$ is a 4-dimensional Minkowski space and take $\Sigma$ to be the hyperplain $x^0=0$. We use $x^i,\theta^i$, $i=1,2,3$ to denote standard coordinates on $T[1]\Sigma$. The induced symplectic potential reads explicitly as:
\begin{equation}
\chi_\Sigma=\tr (F^{0i}\epsilon_{ijk}\theta^j\theta^k\, dC)\,.
\end{equation}
The form $\st{3}{\bar\omega}{}_\Sigma^v$ induced by $\omega_\Sigma=d\chi_\Sigma$ on $SJ^\infty(E_\Sigma)$ (see Section~\bref{sec:pPDE}) reads explicitly as
\begin{equation}
\st{3}{\bar\omega}{}_\Sigma^v=2\tr \left(d\st{0}{F}{}^{0i}dA_i+d\st{1}{F}{}^{0i}_i d\st{0}{C}\right)\theta^{(3)}\,,
\end{equation}
where we introduced coordinates $\psi^A(\theta)$ via
\begin{equation}
\begin{aligned}
F^{ab}(\theta)&=\st{0}{F}{}^{ab}+\st{1}{F}{}^{ab}_j\theta^j+\half\st{2}{F}{}^{ab}_{jk}\theta^j\theta^k\,,\\
C(\theta)&=\st{0}{C}+A_j\theta^j+\half\st{2}{C}_{jk}\theta^j\theta^k\,.
\end{aligned}
\end{equation}
This gives a presymplectic structure $\omega^{\bar E_\Sigma}$ on the fibers of the bundle $\bar E_\Sigma \to X$, defined in such a way that $SJ^{\infty}(E_\Sigma)_0=J^\infty(\bar E_\Sigma)$. 
Furthermore, all the coordinates save for $\cC=\st{0}{C}$, $A_i$, $\pi^i=\st{0}{F}{}^{0i}$ and $\cP=\st{1}{F}{}^{0i}_i$ are in the kernel and the symplectic reduction gives us the reduced bundle $\bar \cE_\Sigma$, equipped with the  canonical symplectic structure
\begin{equation}
\label{can-YM-symp}
\omega^{\bar \cE_\Sigma} = 2 \tr(d\pi^idA_i+d\cC d \cP)(dx)^3\,,
\end{equation}
where we switched to the conventions of Section~\bref{sec:lagBVjet}.  It is easy to see that the space of sections of $\bar\cE_\Sigma$ coincides with the phase space of the constrained Hamiltonian system formulation of Yang-Mills theory.  Indeed, $\pi^i$ are momenta conjugate to $A_i$ and $\cC$ and $\cP$ are canonically conjugated ghosts and ghost momenta. Furthermore,~\eqref{can-YM-symp} induces a canonical symplectic structure on the space of sections of~$\bar\cE_\Sigma$.

Applying~\eqref{aksz-action}  in the present situation  (and having in mind that the degree of $\omega_\Sigma$ is different) we get
\begin{equation}
 \Omega[A,\pi,\cC,\cP]= \int_{T[1]\Sigma}(\pi^*(\chi_\Sigma)(d_\Sigma)+\pi^*(\cL))=
 -2\int d^3x \tr(\pi^i(\d_i\cC+\commut{A_i}{\cC})-\cP \commut{\cC}{\cC})\,,
\end{equation}
which is proportional to the standard expression for the BFV-BRST charge of YM theory. Obtaining BRST-invariant Hamiltonian is not so straightforward in this case because it is not 1-st order in spatial derivatives and hence its derivation involves extra steps. In fact, "covariant Hamiltonian" $\hL$ defined on the initial gauge PDE encodes both the BFV-BRST charge and the covariant analog of the usual Hamiltonian.

\section*{Acknowledgments}
\label{sec:Aknowledgements}
I wish to thank  I.~Dneprov, K.~Druzhkov, J.~Frias, M.~Henneaux, I.~Krasil'shchik, A.~Sharapov, E.~Skvortsov, A.~Verbovetsky for fruitful discussions. I am garateful to A.~Kotov for the collaboration at the early stage of this project.
Part of this work was done when the author participated in the thematic program "Higher Structures and Field Theory" at the Erwin Schrödinger International Institute for Mathematics and Physics, Vienna, Austria. This work was partially completed while the author participated in  the workshop "Higher Spin Gravity and its Applications" supported by the Asia Pacific Center for Theoretical Physics, Pohang, Korea.

\appendix

\section{Technical details}

\subsection*{Descent "covariant Hamiltonians"}
\label{app:BVjet}
In the setting of Section~\bref{sec:lagBVjet} consider $\dv(i_si_s \st{n}{\omega})$. We have
\begin{equation}
\label{des3}
\dv(i_s i_s \st{n}{\omega})=-2i_s L_s \st{n}{\omega}=2i_s\dh \st{n-1}{\omega}=2\dh i_s\st{n-1}{\omega} \,.
\end{equation}
Observe that $i_s\st{n-1}{\omega}$ satisfies: $\dv i_s\st{n-1}{\omega}=-L_s\st{n-1}{\omega}=\dh \st{n-2}{\omega}$ where in the last equality we made use of the descent equation. It follows from triviality of $\dv$ mod $\dh$ cohomology in this degree that
$i_s\st{n-1}{\omega}=-\dv \st{n-1}{\hL}-\dh \st{n-1}{\chi}$ for some $(n-1,0)$-form$ \st{n-1}{\hL}$ and $(n-1,1)$ form $\st{n-1}{\chi}$. Using this in~\eqref{des3} gives
\begin{equation}
    \dv(i_s i_s \st{n}{\omega})=2\dv\dh \st{n-1}{\hL}
\end{equation}
which in turn implies $(\half i_s i_s \st{n}{\omega}-\dh \st{n-1}{\hL})=\pi^*f$  for some $f(x,dx)$. However because $i_s i_s \st{n}{\omega}$ is of degree $n+1$ such $f$ should vanish, giving~\eqref{master}.  

\subsection*{Checking $i_Qi_Q\omega=0$}
\label{app:qqo}
In the setup of definition~\bref{def:presymppde} let $\omega$ be a fixed representative of the vertical form $[\omega]$. Assume in addition that the representative is such that 
\begin{equation}
\label{extra}
i_Q L_Q \omega \in \cI\,, \qquad i_Q d\omega \in \cI\,, \qquad i_Q i_Q d\omega \in \cI\,.
\end{equation}
Let $\hL$ be a covarinat Hamiltonian defined through $i_Q \omega+d\hL \in \cI$. One has
\begin{equation}
d(Q\hL)=d(i_Q d\hL)=L_Q (d\hL)=L_Q(-i_Q \omega+\cI)\in \cI  \,. 
\end{equation}
Because $d(Q\hL)$ is a 1-form $d(Q\hL) \in \cI$ implies $d(Q\hL)=d\pi^*(f)$, $f\in \cC^\infty (T[1]X)$.  However, $\gh{Q\hL}=n+1$ but on $T[1]X$ there are no functions of degree $n+1$ so that $Q\hL=0$.

Furthermore,
\begin{equation}
d (i_Q i_Q \omega)= i_Q d i_Q \omega -L_Q i_Q\omega=
i_Qi_Q d \omega-i_Q L_Q \omega-i_{\commut{Q}{Q}}\omega-i_QL_Q\omega\in \cI\,,
\end{equation}
where we made use of~\eqref{extra} and $Q^2=0$. Just like in deriving $Q\hL=0$ we use $\gh{i_Q i_Q \omega}=n+1$ to see that, again, $i_Q i_Q \omega=0$. Note that if we impose $i_Q i_Q \omega=0$
in place of $Q^2=0$ the last relation implies $i_{Q^2}\omega \in \cI$.

\bibliographystyle{utphys}
\setlength{\itemsep}{5pt}
\small
\providecommand{\href}[2]{#2}\begingroup\raggedright\endgroup

\end{document}